\documentclass{article}

\usepackage{microtype}
\usepackage{graphicx}
\usepackage{subcaption}
\usepackage{booktabs}
\usepackage{array}

\usepackage{tikz}
\usetikzlibrary{calc}
\usetikzlibrary{graphs.standard}
\usetikzlibrary{arrows.meta}

\usepackage[preprint]{icml2026}

\usepackage{amsmath}
\usepackage{amssymb}
\usepackage{mathtools}
\usepackage{amsthm}
\usepackage{thmtools}

\theoremstyle{plain}
\newtheorem{theorem}{Theorem}[section]

\newtheorem{lemma}[theorem]{Lemma}

\theoremstyle{definition}
\newtheorem{definition}[theorem]{Definition}

\theoremstyle{remark}

\newtheorem{example}[theorem]{Example}

\usepackage{hyperref}

\usepackage[capitalize,noabbrev]{cleveref}

\usepackage{natbib}
\usepackage{xurl}

\newenvironment{noseptemize}
{
	\vspace{-2\topsep}
\begin{itemize}
	\setlength{\itemsep}{1pt}
	\setlength{\parskip}{0pt}
	\setlength{\parsep}{0pt}
}{
	\end{itemize}
	\vspace{-2\topsep}
}

\newcommand{\citeauthoryearpar}[1]{\citeauthor{#1} \citeyearpar{#1}}
\newcommand{\snug}[1]{\hspace*{-4em}#1\hspace*{-4em}}

\newcolumntype{V}{>{$}w{c}{1em}<{$}}
\newcolumntype{S}{>{$}c<{$}}

\newcommand{\Natural}{\mathbb{N}}
\newcommand{\Integers}{\mathbb{Z}}
\newcommand{\Rational}{\mathbb{Q}}
\newcommand{\Real}{\mathbb{R}}
\newcommand{\veca}{\overline{a}}
\newcommand{\vecx}{\overline{x}}
\newcommand{\vecy}{\overline{y}}
\newcommand{\vecz}{\overline{z}}
\newcommand{\vecu}{\overline{u}}
\DeclareMathOperator{\id}{id}

\newcommand{\ingate}[1]{x_{#1}}
\newcommand{\outgate}[1]{y_{#1}}
\newcommand{\neighbourin}{N_{\text{in}}}

\DeclareMathOperator{\size}{size}
\DeclareMathOperator{\depth}{depth}

\newcommand{\family}[1]{({#1}_n)_{n \in \Natural}}
\newcommand{\sd}[3]{\textbf{\upshape FSizeDepth}_{#1}(#2, #3)}
\newcommand{\BigO}{\mathcal{O}}

\newcommand{\FCirc}[1]{\textbf{\upshape\scshape FCirc}_{#1}}
\newcommand{\FNC}[2]{\textbf{\upshape FNC}^{#2}_{#1}}
\newcommand{\FSAC}[2]{\textbf{\upshape FSAC}^{#2}_{#1}}
\newcommand{\FAC}[2]{\textbf{\upshape FAC}^{#2}_{#1}}

\newcommand{\tinput}[1][t]{{#1}_{\textit{in}}}
\newcommand{\toutput}[1][t]{{#1}_{\textit{out}}}
\newcommand{\tconst}[1][t]{{#1}_{\textit{c}}}
\newcommand{\tplus}[1][t]{{#1}_+}
\newcommand{\ttimes}[1][t]{{#1}_\times}

\newcommand{\charfin}[2]{\chi^{#1}_{#2}}
\newcommand{\fpoly}[2]{p^{#1}_{#2}}

\newcommand{\finput}{f^{\text{in}}}
\newcommand{\positional}{f^{\text{pos}}}
\newcommand{\attention}{f^{\text{att}}}
\newcommand{\pooling}{f^{\text{pool}}}
\newcommand{\activation}{f^{\text{act}}}

\newcommand{\dpa}{\textbf{\upshape DPA}_R}
\newcommand{\WS}[1]{\textbf{\upshape WS}_R(#1)}
\newcommand{\WP}[1]{\textbf{\upshape WP}_R(#1)}

\DeclareMathOperator{\softmax}{softmax}
\DeclareMathOperator{\avg}{avg}
\DeclareMathOperator{\hardleft}{left}
\DeclareMathOperator{\hardright}{right}

\DeclareMathOperator{\zero}{zero}
\DeclareMathOperator{\sign}{sign}

\newcommand{\encoding}{E_C}

\newcommand{\Vs}[1]{#1_{s}}
\newcommand{\Vp}[1]{#1_{p}}
\newcommand{\Vi}[1]{#1_{i}}
\newcommand{\Vt}[1]{#1_{t}}
\newcommand{\Vv}[1]{#1_{v}}

\newcommand{\Vone}[1]{#1_{one}}
\newcommand{\Vssq}[1]{#1_{ssq}}
\newcommand{\Visq}[1]{#1_{isq}}

\newcommand{\Vbin}[1]{#1_{b}}

\newcommand{\spit}{s, p, i, t}
\newcommand{\spiton}{s, p, i, t, n, one, ssq}

\newcommand{\xin}{\vecx^{in}}
\newcommand{\xplus}{\vecx^+}
\newcommand{\xtimes}{\vecx^\times}

\newcommand{\inputIdentity}{\id}

\newcommand{\transform}[6][\{0\}]{\textbf{\upshape GT}_{#2} \big(#3,\allowbreak #1,\allowbreak #4,\allowbreak #5,\allowbreak #6 \big)}
\newcommand{\transfunc}[6][\{0\}]{\textbf{\upshape FT}_{#2} \big(#3,\allowbreak #1,\allowbreak #4,\allowbreak #5,\allowbreak #6 \big)}

\newcommand{\mult}{\times}

\icmltitlerunning{Average Attention Transformers and Arithmetic Circuits}

\begin{document}

\twocolumn[
	\icmltitle{Average Attention Transformers and Arithmetic Circuits}

	\icmlsetsymbol{equal}{*}

	\begin{icmlauthorlist}
		\icmlauthor{Lena Ehrmuth}{thi}
		\icmlauthor{Laura Strieker}{thi}
	\end{icmlauthorlist}

	\icmlaffiliation{thi}{
		Institute for Theoretical Computer Science, 
		Leibniz University Hanover,
		Hanover, Germany
	}

	\icmlcorrespondingauthor{Lena Ehrmuth}{ehrmuth@thi.uni-hannover.de}
	\icmlcorrespondingauthor{Laura Strieker}{strieker@thi.uni-hannover.de}

	\icmlkeywords{Transformer, Circuit, Circuit Complexity}

	\vskip 0.3in
]

\begin{NoHyper}
	\printAffiliationsAndNotice{}
\end{NoHyper}

\begin{abstract}
	We analyse the computational power of transformer encoders as sequence-to-sequence functions on vectors. 
	We show that average hard attention can be used to simulate arithmetic circuits if they are given as an input to an encoder. 
	The circuit families that can be simulated this way have constant depth while using 
	unbounded addition, binary multiplication and $\sign$ gates. 
	The transformers we use have arithmetic circuits instead of feed-forward networks. 
	With typical average attention the functions they compute are also computed by the same class of circuit families. 
	Our results hold for transformers over the reals, rationals and any ring in between the two. 
\end{abstract}

\section{Introduction}

Transformers \citep{attention_is_all_you_need} have recently gained popularity in machine learning applications as well as theoretical research. 
The analysis of their expressivity has mostly been focused on what decision problems can be solved by transformers. 
In that regard, Boolean circuits have been a promising model for comparison. 
For example, 
\citeauthoryearpar{hao-etal-2022-formal} showed that leftmost hard attention transformers over $\Rational$ 
can only recognize languages in $\textbf{\upshape AC}^0$, 
	the class of functions computed by Boolean circuit families of constant depth and polynomial size with unbounded fan-in. 
\citeauthoryearpar{strobl2023} showed that average hard attention transformers using floats
can only recognize languages in logspace-uniform $\textbf{\upshape TC}^0$ and 
\citeauthoryearpar{merrill2023} showed the same result for soft attention transformers with logarithmic precision. 
There are also a number of descriptive characterizations comparing transformers to logics. 
\citeauthoryearpar{Chiang2023} defined an extension of first-order logic 
with modular predicates on positions and counting quantifiers and 
showed that its definable languages are between the languages recognized by 
soft attention transformers with fixed and unlimited precision. 
\citeauthoryearpar{Barcelo2024} showed that all languages definable in 
first-order logic with arbitrary unary predicates can be recognized by 
leftmost hard attention transformers and 
even when extending this logic with counting terms, 
the definable languages are recognized by average hard attention transformers. 
Whereas \citeauthoryearpar{Yang2024} showed that transformers 
using left- and rightmost 
attention with strict masking and no positional embedding 
recognize exactly the languages definable in linear temporal logic. 

All of those results are limited to decision problems over finite alphabets. 
We want to analyse the expressivity of transformers as functions 
from vector sequences to vector sequences and compare them to arithmetic circuits. 
These circuits take numbers as inputs and their gates compute functions like addition, multiplication 
and the $\sign$ function. 
This allows us to reason precisely about the theoretical transformer model, which uses real numbers and 
also about practical implementations using rational numbers for approximation. 
We adapt a model of generalized transformers by \citeauthoryearpar{hao-etal-2022-formal} and 
use it to define a class of average dot-product attention transformers 
which use arithmetic circuits instead of feed-forward networks. 

\paragraph{Main Results}
The main contribution of this paper is a correspondence between average attention circuit transformers and the function class $\FSAC{R}{0}[\sign]$. 
All functions computed by the transformers are in $\FSAC{R}{0}[\sign]$ and 
all circuits from an $\FSAC{R}{0}[\sign]$ family can be simulated by a single such transformer. 

Our results give a new perspective on the expressivity of transformers and 
we think that the field of machine learning would benefit from advances in the theory of algebraic computation and vice versa. 
We would also like to emphasize that theoretical results have the potential to inspire practical development. 

\paragraph{Organization}
In the next section, we define arithmetic circuits and generalized transformers over semirings. 
In \cref{sec:simulating_transformers}, we briefly describe an upper bound for the circuit complexity of transformers. 
In \cref{sec:simulating_circuits}, we show how average attention circuit transformers can be used to simulate circuits. 
This requires a rather intricate construction presented incrementally in multiple subsections. 
The full proofs of most of the theorems are deferred to the Appendix. 

\section{Preliminaries}
\label{sec:preliminaries}

We define $\Natural = \{0, 1, 2, \dots\}$. 
For a set $S$, we use $S^*$ for the Kleene closure of $S$. 
We use lowercase letters to denote numbers, overlines for vectors 
and uppercase letters for sequences of vectors. 
The length of a vector or sequence is denoted with $\vert\vecx\vert$ or $\vert X \vert$. 
For $d \in \Natural$, the components of a vector $\vecx \in S^d$ are $x_1, \dots, x_d$ 
and the vectors of a sequence $X \in (S^d)^*$ are $X_1, \dots, X_{\vert X \vert}$. 

\subsection{Algebra}

We will discuss arithmetic circuits and transformers using values from arbitrary semirings. 
\begin{definition}
	A \emph{semiring} is a set $R$ equipped with binary functions $+$ and $\mult$ such that 
	\begin{noseptemize}
		\item $(R, +)$ is a commutative monoid with identity $0$, 
		\item $(R, \mult)$ is a monoid with identity $1$, 
		\item $0$ is absorbing with respect to $\mult$ and
		\item $\mult$ distributes over $+$. 
	\end{noseptemize}
	A semiring is called \emph{commutative} if $(R, \mult)$ is a commutative monoid. 
	A \emph{ring} is a semiring such that $(R, +)$ is an abelian group. 
	A \emph{field} is a ring such that $(R, \mult)$ is an abelian group. 
\end{definition}
We can perform addition and multiplication over any semiring. 
A ring also allows subtraction and over a field we can additionally use division. 
For example, 
$\Natural$ is a semiring but not a ring, 
$\Integers$ is a ring but not a field and 
$\Rational, \Real$ are fields. 

If we want to simulate the computation of one semiring within another, 
we need what we call an extension. 
\begin{definition}
	Let $R$ and $T$ be semirings with functions $+_R, \mult_R$ and $+_T, \mult_T$ 
	having identities $0_R, 1_R$ and $0_T, 1_T$. 
	Then a function $f \colon R \to T$ is a \emph{homomorphism} if we get 
	\begin{align*}
		f(x +_R y) &= f(x) +_T f(y), \\
		f(x \mult_R y) &= f(x) \mult_T f(y), \\ \text{and }
		f(1_R) &= 1_T 
	\end{align*}
	for all $x, y \in R$. 
	We call $T$ a \emph{semiring extension} of $R$ if there is an injective homomorphism from $R$ to $T$. 
	If $T$ is also a ring/field, we call $T$ a \emph{ring/field extension} of $R$. 
\end{definition}
For example,
$\Natural$ is a semiring extension of $\Natural$, 
$\Integers$ is a ring extension of $\Natural$ and 
$\Rational$ is a field extension of $\Integers$. 

\begin{definition}
	An ordered semiring is a semiring $R$ with a partial order $\mathord\le \subseteq R^2$ 
	such that for all $x, y, z \in R$ we get 
	\begin{align*}
		x \le y \:&\Rightarrow\: x + z \le y + z \quad\text{and} \\
		x \le y \text{ and } 0 \le z \:&\Rightarrow\: x \mult z \le y \mult z \text{ and } z \mult x \le z \mult y .
	\end{align*}
	We call a ring extension of $\Integers$ $\Integers$-ordered, 
	if it preserves the usual order of $\Integers$, which means that for all $x, y \in \Integers$ 
	\[f(x) \le f(y) \quad\Leftrightarrow\quad x \le y.\]
\end{definition}
We will need an order on our values to define hard attention. 

\subsection{Arithmetic Circuits}
\label{sec:arithmetic_circuits}

We compare transformers to the well established theoretical model of circuits. 
For a more thorough introduction to this topic, see \citeauthoryearpar{Savage1997} or \citeauthoryearpar{Vollmer1999}. 
A circuit is defined by an acyclic graph labeled with functions over some set. 
In this paper, this will be a semiring with its operations $+$ and $\mult$. 

\begin{definition}
	A \emph{circuit basis} is a pair $(S, B)$ where $S$ is a set or domain and $B$ is a set of 
	functions $f \colon S^k \to S$ with $k \in \Natural$ and 
	functions on sequences $f \colon S^* \to S$. 
\end{definition}

\begin{definition}
	A \emph{circuit} $C$ over a basis $(S, B)$ with $m$ inputs and $n$ outputs 
	is a tuple $C = (V, E, \alpha, \beta)$ where 
	\begin{noseptemize}
		\item $(V, E)$ is a finite directed acyclic graph, 
			the nodes of which are called the gates of $C$, 
		\item $\beta \colon V \to S \cup B \cup \{\ingate{1}, \dots, \ingate{m}, \outgate{1}, \dots, \outgate{n}\}$ 
			is a labeling function that assigns each element of $\{\ingate{1}, \dots, \ingate{m}, \outgate{1}, \dots, \outgate{n}\}$ exactly once and
		\item $\alpha \colon E \to \Natural$ is a function such that for each node $v \in V$, 
			$\alpha(\neighbourin(v)) = \{1, \dots, \vert \neighbourin(v) \vert\}$ where 
			$\neighbourin(v) = \{u \mid (u, v) \in E\}$ is the in-neighbourhood of $v$. 
	\end{noseptemize}
	We call the in- and out-degree of gates their fan-in and fan-out. 
	The function $\beta$ determines the type of each gate 
	and thereby imposes some conditions on the graph: 
	\begin{noseptemize}
		\item A gate $v$ with $\beta(v) \in S$ is called a constant gate and has fan-in $0$. 
		\item A gate $v$ with $\beta(v) \in \{\ingate{1}, \dots, \ingate{m}\}$ is called an input gate and has fan-in $0$. 
		\item A gate $v$ with $\beta(v) \in \{\outgate{1}, \dots, \outgate{n}\}$ is called an output gate, 
			has fan-in $1$ and fan-out $0$.
		\item A gate $v$ with $\beta(v) \in B$ is called a function gate and has no restrictions.  
	\end{noseptemize}
	The \emph{size} of a circuit is the number of gates it has and 
	the \emph{depth} of a circuit is the length of the longest path from an input to an output gate. 
\end{definition}

\begin{definition}
	A circuit $C = (V, E, \alpha, \beta)$ over $(S, B)$ with $m$ inputs and $n$ outputs 
	computes a function $f_C \colon S^m \to S^n$ defined as follows. 
	Each gate $v \in V$ computes a function $f_{C,v} \colon S^m \to S$ inductively defined as
	\begin{equation*}\vecx \mapsto 
		\begin{cases}
			c		&\text{if } \beta(v) = c \in S \\
			x_i	&\text{if } \beta(v) = \ingate{i} \\
			f_{C,u}(\vecx) &
			\begin{multlined}
				\text{if } \beta(v) \in \{\outgate{1}, \dots, \outgate{n}\} \\[-11pt]
				\text{ and } \neighbourin(v) = \{u\} 
			\end{multlined}
				\\[2pt]
			\begin{multlined}
				b(f_{C,u_1}(\vecx), \dots \\[-11pt] \dots, f_{C,u_k}(\vecx)) \\
			\end{multlined} &
			\begin{multlined}
				\text{if } \beta(v) = b \in B \text{ and} \\[-11pt] \hspace{1pt}
				\neighbourin(v) = \{u_1, \dots, u_k\} \\ \text{ s.t. }
				\alpha(u_1) < \dots < \alpha(u_k)
			\end{multlined}
		\end{cases}
	\end{equation*}
	and $f_C(\vecx) = (f_{C, \beta^{-1} (\outgate{1})}, \dots, f_{C, \beta^{-1} (\outgate{n})})$. 
\end{definition}
A single circuit takes an input vector whose dimension is equal to the number of input gates. 
To handle inputs sequences of any length, we need a family of circuits. 

\begin{definition}
	A \emph{circuit family} $C = \family{C}$ over $(S, B)$ is a sequence of circuits over $(S, B)$ 
	such that each circuit $C_n$ has exactly $n$ inputs. 
	It computes a function \[f_C \colon S^* \to S^*, \vecx \mapsto f_{C_{\vert\vecx\vert}} (\vecx).\] 
	
	We define $\sd{(S, B)}{f_s}{f_d}$ as the set of functions computed by circuit families over $(S,B)$ 
	that have $\size(C_n) \in \BigO(f_s(n))$ and $\depth(C_n) \in \BigO(f_d(n))$. 
\end{definition}

If the basis is a semiring with addition and multiplication functions, 
the corresponding circuits are called arithmetic. 
We define three different complexity classes based on the fan-in of addition and multiplication gates. 
\begin{definition}
	Let $R$ be any semiring with operations $+$ and $\mult$. 
	We define $+_* \colon R^* \to R, \vecx \mapsto \sum_{i=1}^{\vert\vecx\vert} x_i$ and 
	$\mult_* \colon R^* \to R, \vecx \mapsto \prod_{i=1}^{\vert\vecx\vert} x_i$. 
	Then a bounded fan-in (arithmetic) circuit over $R$ is a circuit over the basis $(R, \{+, \mult\})$, 
	a semi-unbounded fan-in circuit over $R$ uses the basis $(R, \{+_*, \mult\})$ and 
	an unbounded fan-in circuit over $R$ uses the basis $(R, \{+_*, \mult_*\})$. 
	We use $\FCirc{R}$ for functions computed by an unbounded fan-in circuit over $R$. 
	These always have fixed size input and output vectors. 
	To analyse functions of sequences, we define complexity classes for the three kinds of circuit families. 
	\begin{align*}
		\FNC{R}{0} &= \sd{(R, \{+, \mult\})}{n^{\BigO(1)}}{1} \\
		\FSAC{R}{0} &= \sd{(R, \{+_*, \mult\})}{n^{\BigO(1)}}{1} \\
		\FAC{R}{0} &= \sd{(R, \{+_*, \mult_*\})}{n^{\BigO(1)}}{1} 
	\end{align*}
\end{definition}
All three classes restrict circuit families to polynomial size and constant depth, 
but we get $\FNC{R}{0} \subsetneq \FSAC{R}{0} \subsetneq \FAC{R}{0}$. 
For example, $(\sum_{i=1}^n x_i)_{n \in \Natural} \in \FSAC{R}{0} \setminus \FNC{R}{0}$ and 
$(\prod_{i=1}^n x_i)_{n \in \Natural} \in \FAC{R}{0} \setminus \FSAC{R}{0}$, 
because using only binary gates they require circuit families of logarithmic depth. 

Because these circuits can only compute functions where each component is defined by a polynomial, 
which are continuous, while transformers can compute uncontinuous functions, 
we need to extend the circuit bases with other functions. 
For this purpose we use the following notation. 
\begin{definition}
	Let $(S, B)$ and $(S, A)$ be circuit bases. 
	Then $\sd{(S, B)\!}{f_s}{\!f_d}[A] = \sd{(S, A \cup B)\!}{f_s}{\!f_d}$. 
	For a semiring $R$, we use $\FCirc{R}[A]$ for functions computed by circuits over 
	$(R, \{+_*, \mult_*\} \cup A)$. 
\end{definition}
When explicitly listing the contents of the set $A$, the braces will be omitted, 
e.g. $\FNC{R}{0}[f, g]$ instead of $\FNC{R}{0}[\{f, g\}]$. 

\subsection{Characteristic Functions}

It will later be useful to distinguish elements of a fixed finite set. 
This requires the following kind of function. 
\begin{definition}
	Let $A$ be a set and $a \in A$. 
	Then the function $\charfin{A}{a} \colon A \to \{0,1\}$ is defined as 
	\[\charfin{A}{a} (x) = \begin{cases} 1 &\text{if } x = a \text{ and} \\ 0 &\text{otherwise}\end{cases}.\]
    \vspace{-12pt}
\end{definition}
Such a function is not generally computable by an arithmetic circuit unless $A = \{0, 1\}$. 
One way to mitigate this fact is adding the $\zero$ function to a ring. 
It is defined as \[\zero(x) = \begin{cases} 1 &\text{if } x = 0 \\ 0 &\text{otherwise} \end{cases}. \]
\vspace{-8pt}
\begin{lemma} \label[lemma]{thm:charfin_zero}
	Let $R$ be a ring and $r \in R$. 
	Then $\charfin{R}{r} \in \FCirc{R}[\zero]$. 
\end{lemma}
\begin{proof}
	The following formula 
	describes such a circuit: 
	\[\charfin{R}{r}(x) = \zero(x - r) = \zero(x + (-1) \mult r) \qedhere\] 
\end{proof}
If $R$ is a field, we can instead use Lagrange polynomials to interpolate any finite set of values.
\begin{restatable}{lemma}{lagrange} \label{thm:charfin_field}
	Let $R$ be a field, $A \subseteq R$ be finite and $a \in A$. 
	Then there is a polynomial $\fpoly{A}{a} \colon R \to R$ in $\FCirc{R}$ 
	such that $\fpoly{A}{a}(x) = \charfin{A}{a}(x)$ for all $x \in A$. 
\end{restatable}
The proof can be found in the \hyperref[app:lagrange]{Appendix}. 

\subsection{Generalized Transformers}
\label{sec:generalized_transformers}

We adapt a definition by \citeauthoryearpar{hao-etal-2022-formal}
to describe transformer encoders using values from arbitrary semirings. 
Also, instead of using them for decision problems, 
we look at the functions computed by encoders with no output layer. 
This allows us to directly output real values and compare them to the values computed by circuit families. 

A transformer can have multiple layers with multiple attention heads. 
Each head uses an attention function to compute attention scores between pairs of vectors and 
a pooling function to combine all vectors based on those. 
The activation function of a layer takes in the results of its attention heads. 
Usually, the activation functions are computed by feedforward networks, 
the attention functions are dot-products and 
the pooling functions apply softmax to the attention scores and compute a weighted sum of all vectors. 

\begin{definition}
    A \emph{generalized transformer} with $K$ layers and $H$ attention heads is a tuple 
    $(R,\allowbreak d,\allowbreak \Sigma,\allowbreak \finput,\allowbreak \positional,\allowbreak 
		\attention_{1,1},\allowbreak \attention_{1, 2}, \dots, \attention_{K,H},\allowbreak 
		\pooling_{1,1}, \dots, \pooling_{K,H},\allowbreak
        \activation_1, \dots, \activation_K)$ where
    \begin{noseptemize}
		\item $R$ is the underlying semiring, 
		\item $d \in \Natural$ is the dimension, 
		\item $\Sigma$ is the (possibly infinite) input alphabet, 
        \item $\finput\colon \Sigma \to R^d$ is the input embedding, 
        \item $\positional\colon \Natural \times \Natural \to R^d$ is the positional embedding, 
        \item $\attention_{k,h}\colon R^d \times R^d \to R$ is the attention function for attention head $h$ in layer $k$,
        \item $\pooling_{k,h}\colon (R^d)^* \times R^* \to R^d$ is the pooling function for attention head $h$ in layer $k$ and
		\item $\activation_k\colon (R^d)^{H+1} \to R^d$ is the activation function for layer $k$.
    \end{noseptemize}
	Such a transformer computes a function $T \colon \Sigma^* \to (R^d)^*$ defined by induction over $k$. 
	First, the input sequence $X = (x_1, \dots, x_n) \in \Sigma^*$ is turned into 
	a sequence of $R^d$-vectors using input and positional embedding: 
	\[Y^0 = \big(\finput(x_1) + \positional(1, n), \dots, \finput(x_n) + \positional(n, n) \big)\]
	Then each layer $k$ transforms the previous sequence $Y^{k-1}$ into a sequence $Y^k$ of the same length in the following way. 
	The attention function of each attention head $h$ produces an $n \times n$ matrix of attention values: 
	\[a^{k,h}_{i,j} = \attention_{k,h} (Y^{k-1}_i, Y^{k-1}_j) \quad\text{for } 1 \le i,j \le n\]
	The pooling function of each attention head $h$ turnes it back into a vector sequence: 
	\[Z^{k,h}_i = \pooling_{k,h} \big(Y^{k-1}, (a^{k,h}_{i,1}, \dots, a^{k,h}_{i,n}) \big) \quad\text{for } 1 \le i \le n\]
	Finally, the activation function of the layer is applied to each vector in those sequences: 
	\[Y^k_i = \activation_k \big(Y^{k-1}_i, Z^{k, 1}_i, \dots, Z^{k, H}_i \big) \quad\text{for } 1 \le i \le n\]
	The output of the transformer is just the output of the last layer: $T(X) = Y^K$. 
\end{definition}

\subsection{Restricted Transformers}
\label{sec:restricted_transformers}

We can define classes of transformers by restricting the choices for the functions they use. 
\begin{definition}
	Let $R$ be a semiring and $A, B, C, D, E$ be sets of functions. 
	Then $\transform[B]{R}{A}{C}{D}{E}$ is the set of transformers using 
	$R$ as the underlying semiring, input embeddings from $A$, 
	positional embeddings from $B$, activation functions from $C$, 
	attention functions from $D$ and pooling functions from $E$ and 
	$\transfunc[B]{R}{A}{C}{D}{E}$ is the set of functions computed by them. 
\end{definition}
To describe the original transformer model by \citeauthoryearpar{attention_is_all_you_need}, 
we define three sets of functions. 
$\dpa$ is the set of attention functions computable via dot-products. 
$\WS{f}$ describes a weighted sum using the resulting values of a function $f$ on the attention scores. 
And $\WP{f}$ is a variation of $\WS{f}$ that uses multiplication instead of addition 
while also excluding zero values and will be used later. 

\begin{definition}
	Let $R$ be a semiring and $f\colon R^* \to R^*$. 
	\allowdisplaybreaks
	\begin{align*}
		\dpa &= \left\{
		\begin{gathered}
			g \colon R^d \times R^d \to R, \\
			(\vecx, \vecy) \mapsto A \vecx \cdot B \vecy
		\end{gathered}
		\,\middle|\, 
		\begin{gathered}
			d \in \Natural \text{ and} \\ A, B \in R^{d \times d}
		\end{gathered}
		\right\}
	\\
		\WS{f} &= \left\{
		\begin{gathered}
			g \colon (R^d)^* \times R^* \to R^d, \\
			(X, \veca) \mapsto \sum_{i=1}^{\vert X \vert} f(\veca)_i \mult X_i 
		\end{gathered}
		\,\middle|\, d \in \Natural
		\right\}
	\\
		\WP{f} &= \left\{
		\begin{gathered}
			g \colon (R^d)^* \times R^* \to R^d, \\
			(X, \veca) \mapsto \bigodot_{\substack{i=1 \\ \snug{f(\veca)_i \neq 0}}}^{\vert X \vert} 
			f(\veca)_i \mult X_i
		\end{gathered}
		\,\middle|\, d \in \Natural
		\right\}
	\end{align*}
	Here $\odot$ is used for element-wise multiplication of vectors. 
\end{definition}
The original model uses $R = \Real$ 
and pooling functions from $\WS{\softmax}$ where 
\[\softmax(\veca)_i = \frac{e^{a_i}}{\sum_{j = 1}^{\vert \veca \vert} e^{a_j}}.\]
This scheme is called soft attention. 
In this paper, we are only concerned with hard attention, 
because it makes it possible to attend only to specific vectors. 
There are three variants of hard attention. 
Leftmost (or unique) and rightmost attention are defined by the functions 
\begin{align*}
	\hardleft(\veca) &= 
	\begin{cases}
		1 & \text{if } i = \min M \\
		0 & \text{otherwise}
	\end{cases}
	\\ \text{ and }  
	\hardright(\veca) &= 
	\begin{cases}
		1 & \text{if } i = \max M \\
		0 & \text{otherwise}
	\end{cases}
	\\ \text{ with } M &= \{i \mid a_i \text{ is maximal in } \{a_1, \dots, a_{\vert \veca \vert}\}\} .
\end{align*}
Here we say that each vector attends to only one vector. 
Average attention is defined by the function 
\[
	\avg(\veca)_i = \begin{cases}
		\frac{1}{\vert M \vert} & \text{if } i \in M \\
		0 & \text{otherwise}
	\end{cases} .
\]
Here each vector attends to $\vert M \vert$ vectors. 

These attention schemes correspond to pooling functions from 
$\WS{\hardleft}, \WS{\hardright}$ and $\WS{\avg}$ respectively. 
They require a (partial) order on $R$ to define maximal elements while 
$\avg$ also relies on the existence of multiplicative inverses of $\Natural \setminus \{0\}$.

\section{Simulating Transformers with Circuits}
\label{sec:simulating_transformers}

The equations describing the computation of a transformer can be seen as a description of a circuit family. 
Their length is only dependent on the length $n$ of the input sequence. 
If all functions used by a generalized transformer $T$ are computable by circuits and we fix $n$, 
we can connect those together to construct a circuit for $T$. 
The circuits for all $n$ form a circuit family. 
This requires the input alphabet $\Sigma$ to be a subset of $R^k$ for some $k$, 
but if it is finite, the complexity of a mapping $f \colon \Sigma \to R^k$ is essentially trivial. 

Because transformers have a fixed number of layers, 
the depth of the resulting circuits grows linearly with the depth of the circuits for the functions they use. 
The pooling functions are the only ones that require a circuit family, 
as all others are defined on fixed-size vectors. 
Therefore, if those families have constant depth, so will the family for the transformer. 
The size also grows linearly with the size of the pooling function circuits, 
but it is at least quadratic in $n$, because of the $n \times n$ matrix of attention values. 
For a detailed explanation of this construction, see \citeauthoryearpar{ma_ehrmuth}. 

\begin{theorem}
    Let $R$ be any semiring, $(R, B)$ a circuit basis and $P \subseteq \{f \colon \Natural^2 \to R\}$. Then 
	$\transfunc[P]{R}{\FCirc{R}[B]}{\FCirc{R}[B]}{\FCirc{R}[B]}{\FNC{R}{0}[B]} \allowbreak \subseteq \FNC{R}{0}[B]$. 
	This also holds if we replace $\FNC{}{}$ with $\FSAC{}{}$ or $\FAC{}{}$. 
\end{theorem}
There is no restriction on the positional embedding, 
because it does not actually need to be computed by the circuit family. 
Since it is only dependent on the number of inputs, 
its values can be stored in constant gates. 

\subsection{Simulating Restricted Transformers}

Dot-products over any semiring $R$ can be computed by arithmetic circuits, so $\dpa \subseteq \FCirc{R}$. 
For the hard attention functions, we get $\hardleft, \hardright, \avg \in \FSAC{R}{0}[\sign]$ 
where \[\sign(x) = \begin{cases}
	1 &\text{if } x > 0 \\
	0 &\text{if } x = 0 \\
	-1&\text{if } x < 0 
\end{cases}.\]
So $\WS{\hardleft}, \WS{\hardright}, \WS{\avg} \subseteq \FSAC{R}{0}[\sign]$. 
The most common type of activation function is a feed-forward network using the 
rectified linear unit function which are also part of $\FCirc{R}[\sign]$. 
All together, these restricted transformers compute functions in $\FSAC{R}{0}[\sign]$. 

\section{Simulating Circuits with Transformers}
\label{sec:simulating_circuits}

Because transformers are length-preserving, 
they cannot directly compute functions that are not. 
To work around this difference, 
the input to the transformer must be long enough to fit the output of the function. 
We choose a similar approach to \citeauthoryearpar{GNNs_and_arithmetic_circuits}. 
They used graph neural networks on graphs of circuits to compute their output values. 
Since transformers only work on vectors of numbers, we need some way to encode the circuits as such. 

We use one vector for each edge in the circuit, so the length of the encoding is quadratic with respect to the size of the circuit. 
We assume that the gates are numbered with indices from $\Natural$. 
The types of the gates will be encoded with their incoming edges. 
Since the input and constant gates do not have any, we also add one vector for each of those. 

\newcommand{\Types}{T}

\begin{definition}
	Let $R$ be a semiring extension of $\Natural$, $C = (V, E, \alpha, \beta)$ an $R$-circuit with $n$ inputs 
	and $\vecu \in R^n$. 
	We fix a set $\Types = \{\tconst, \tinput, \toutput, \tplus, \ttimes\} \subseteq R$ 
	of different constants for the gate types. 
	For $V = \{v_1, \dots, v_{\vert V \vert}\}$ we define the \emph{encoding} $\encoding(C, \vecu) \in (R^5)^*$ as follows: 
	\begin{noseptemize}
		\item For each constant gate $v_i \in \{v \mid v \in V, \beta(v) \in R\}$ 
			we add the vector $(i, 0, 0, \tconst, \beta(v_i))$. 
		\item For each input gate $v_i \in \{v \mid \beta(v) \in \{\ingate{1}, \dots, \ingate{n}\}\}$ 
			we add the vector $(i, 0, 0, \tinput, u_k)$ where $\beta(v_i) = \ingate{k}$. 
		\item For each edge $(v_i, v_j) \in E$ we add the vector $(j, i, \alpha((v_i, v_j)), t, 0)$ 
			where $t$ is either $\tplus, \ttimes$ or $\toutput$ 
			depending on whether $v_j$ is a $+, \times$ or output gate. 
	\end{noseptemize}
\end{definition}
To improve the readability, 
we denote the components $(x_1, x_2, x_3, x_4, x_5)$ of these vectors by $(\Vs{x}, \Vp{x}, \Vi{x}, \Vt{x}, \Vv{x})$, 
so the vector $\vecx$ for edge $(u, v)$ stores 
	the index of the successor $v$ in $\Vs{x}$, 
	the index of the predecessor $u$ in $\Vp{x}$, 
	$\alpha((u, v))$ in $\Vi{x}$, 
	the type of $v$ in $\Vt{x}$ and 
	the value assigned to $v$ in $\Vv{x}$. 
\begin{example} \label[example]{example_encoding_unbounded}
	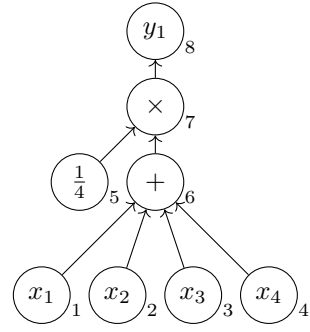
\begin{figure}[t]
		\centering
		\begin{tikzpicture}
			\graph[branch right, grow up, nodes={circle, draw, minimum size=0.75cm, inner sep=0pt}]
			{
				{ [nodes={y=-0.5}] 1/$\ingate{1}$, 2/$\ingate{2}$, 3/$\ingate{3}$, 4/$\ingate{4}$ }
				-> { [nodes={x=1.5}] 6/$+$ -> 7/$\times$ -> 8/$\outgate{1}$ };
			};
			\node (5) [circle, draw, minimum size=0.75cm, inner sep=0pt] at ($(6) - (1, 0)$) {$\frac{1}{4}$};
			\draw[->] (5) -- (7);

			\foreach \i in {1, ..., 8}
				\node at (\i.east) [anchor=north west, inner xsep=0pt] 
					{\scriptsize $\i$};
		\end{tikzpicture}

		\caption{An unbounded fan-in circuit with a numbering for the gates. The edges are numbered from left to right. }
		\label{example_encoding_unbounded_figure}
	\end{figure}
	\begin{table}[t]
		\caption{
			Encoding of the circuit in \cref{example_encoding_unbounded_figure} 
			with input vector $\overline{u}$. 
		}
		\label{tab:example_encoding_unbounded}
		\vskip 0.1in \centering
		\begin{tabular}{VVVVVV}
			\toprule & \Vs{x} & \Vp{x} & \Vi{x} & \Vt{x} & \Vv{x} \\ \midrule
			\vecx_{1}  & 1 & 0 & 0 & \tinput & u_1 			\\
			\vecx_{2}  & 2 & 0 & 0 & \tinput & u_2 			\\
			\vecx_{3}  & 3 & 0 & 0 & \tinput & u_3 			\\
			\vecx_{4}  & 4 & 0 & 0 & \tinput & u_4 			\\
			\vecx_{5}  & 5 & 0 & 0 & \tconst &\frac{1}{4}	\\
			\vecx_{6}  & 6 & 1 & 1 & \tplus  & 0 			\\
			\vecx_{7}  & 6 & 2 & 2 & \tplus  & 0 			\\
			\vecx_{8}  & 6 & 3 & 3 & \tplus  & 0 			\\
			\vecx_{9}  & 6 & 4 & 4 & \tplus  & 0 			\\
			\vecx_{10} & 7 & 5 & 1 & \ttimes & 0 			\\
			\vecx_{11} & 7 & 6 & 2 & \ttimes & 0 			\\
			\vecx_{12} & 8 & 7 & 1 & \toutput& 0 			\\
			\bottomrule
		\end{tabular}
	\end{table}
	The circuit in \cref{example_encoding_unbounded_figure} will be encoded into 
	the vector sequence in \cref{tab:example_encoding_unbounded}.
	The addition gate has predecessors $1, 2, 3, 4$ with each edge index matching its predecessor. 
	The multiplication gate has two edges: 
	The first one from gate $5$ and the second one from gate $6$. 
	The output node has only one predecessor and, therefore, the encoding has only one vector with $\toutput$. 
\end{example}

\begin{definition}
	We say that a transformer $T$ \emph{simulates} a circuit $C$ with $n$ inputs encoded with $\encoding$, 
	if for all $\vecx \in R^n$, 
	the values assigned to the incoming edges of output gates in $T(\encoding(C, \vecx))$ are exactly $f_C(\vecx)$. 
\end{definition}
We can use attention heads to fetch the values of the predecessors of each gate 
and propagate them along the edges. 
This is a relatively simple construction if we leave some freedom for the attention and pooling functions, 
but it is still possible with a restriction to hard dot-product attention. 
The proofs for all statements in this section have been deferred to the appendix. 

\subsection{Using Generalized Attention}
\label{sec:sim_generalized}

For the input embedding we use the identity function $\inputIdentity$, 
so the transformer uses the same dimension as the encoding. 
We define two attention functions $\attention_E, \attention_V \colon R^5 \times R^5 \to R$. 
The first one, \[\attention_E(\vecx, \vecy) = \begin{cases} 1 
	&\text{if } \Vp{x} = \Vs{y} \\ 0 &\text{otherwise} \end{cases},\]
lets $\vecx$ attend only to those vectors describing incoming edges of the $\Vp{x}$-th gate. 
It will be used to transfer the values from the incoming to the outgoing edges of each node. 
The second attention function, \[\attention_V(\vecx, \vecy) = \begin{cases} 1 
	&\text{if } \Vs{x} = \Vs{y} \\ 0 &\text{otherwise} \end{cases},\]
lets the vectors of all incoming edges of the same gate attend only to each other.
It will be used to compute the values of the gates from the values of its predecessors. 
If $R$ is a ring, both functions are computed by an $R$-circuit using the $\zero$ function, 
since $\attention_E(\vecx, \vecy) = \zero(\Vp{x} - \Vs{y})$ and $\attention_V(\vecx, \vecy) = \zero(\Vs{x} - \Vs{y})$. 

Now we can define one pooling function for addition and one for multiplication to simulate the arithmetic gates. 
They will be used by different attention heads in the same layer 
and the activation function will choose which one to use based on the gate type. 
One possible choice for pooling functions is
\[
	\pooling_+(X, \veca) = \sum_{i=1}^{\vert X \vert} a_i \mult X_i \text{ and }
	\pooling_\times(X, \veca) = \bigodot_{\substack{i=1 \\ a_i \neq 0}}^{\vert X \vert} a_i \mult X_i. 
\]
Note that $\pooling_+ \in \WS{\id}$ and $\pooling_\times \in \WP{\id}$ where $\id\colon R^* \to R^*$ is the identity on sequences. 
To choose the right value based on the gate type, we need the functions 
$\charfin{\Types}{t} \colon \Types \to \{0,1\}$ for $t \in \Types = \{\tconst, \tinput, \toutput, \tplus, \ttimes\}$. 
We define an activation function $\activation_V\colon (R^5)^3 \to R^5$ as 
\begin{align*}
	\Vv{\activation_V(\xin, \xplus, \xtimes)} = 
	\big(\charfin{T}{\tinput}(\Vt{\xin}) + \charfin{T}{\tconst}(\Vt{\xin})\big) &\mult \Vv{\xin} + \\
	\big(\charfin{T}{\toutput}(\Vt{\xin}) + \charfin{T}{\tplus}(\Vt{\xin})\big) &\mult \Vv{\xplus} + \\
	\charfin{T}{\ttimes}(\Vt{\xin}) &\mult \Vv{\xtimes} 
\end{align*}
and $\activation_V(\xin, \xplus, \xtimes)_j = \xin_j$ for $j \in \{\spit\}$. 
This function retains the values for input and constant gates as well as all information about the circuit. 
Due to \cref{thm:charfin_zero}, we get $\activation_V \in \FCirc{R}[\zero]$ for any ring $R$, 
but if $R$ is a field, we can use polynomials from \cref{thm:charfin_field} instead. 

Because the attention function $\attention_E$ must be used before this to move the predecessor values into $\Vv{x}$, 
we add a second transformer layer. 
The activation function $\activation_E \colon (R^5)^3 \to R^5$ for this layer is defined as 
\begin{multline*}
	\Vv{\activation_E(\vecx, \vecy, \vecz)} = 
	\big(\charfin{T}{\tinput}(\Vt{x}) + \charfin{T}{\tconst}(\Vt{x})\big) \mult \Vv{x} + \\
	\big(\charfin{T}{\toutput}(\Vt{x}) + \charfin{T}{\tplus}(\Vt{x}) + \charfin{T}{\ttimes}(\Vt{x})\big) 
	\mult \Vv{y} 
\end{multline*}
and also retains the other values. 
For the pooling function we can use $\pooling_+$ or $\pooling_\times$, 
because each vector of a correct encoding attends to at most one vector. 

When we use all these functions to define a transformer 
\vspace{-0.5\baselineskip}
\begin{multline*}
	\Big(R, 5, R^5, \inputIdentity, 0, \attention_E, 0, \attention_V, \attention_V, 
	\\
	\pooling_+, 0, \pooling_+, \pooling_\times, \activation_E, \activation_V \Big)
\end{multline*}
with two layers, it is able to simulate circuits of depth $1$ and 
by stacking more copies of these two layers, we can simulate circuits of higher depth. 
We write $0$ for functions of unused attention heads, 
but $\attention_E$ and $\pooling_+$ could be used instead, 
because their values are ignored by $\activation_E$. 

\begin{restatable}{theorem}{generalized} \label{thm:generalized_id}
	Let $R$ be a commutative ring extension of $\Integers$. 
	Then for each $K \in \Natural$ there is a transformer 
	in $\transform{R}{\{\inputIdentity\}}{\FCirc{R}[\zero]}
	{\FCirc{R}[\zero]}{\WS{\id} \cup \WP{\id}}$ 
	of dimension $5$ with $2K$ layers 
	that simulates any unbounded fan-in $R$-circuit of depth $\le K$ encoded with $\encoding$.
\end{restatable}
This also implies that, for each $\FAC{R}{0}$-circuit family, 
there is a transformer that can simulate all of its circuits.

\subsection{Using Hard Attention}

A similar construction can be done using average dot-product attention, but there are two key differences. 
Firstly it attends only to the maximal attention values. 
Therefore we need two dot-products $\attention_E, \attention_V$ such that 
for a fixed $\vecx$ the values $\attention_E(\vecx, \vecy)$ and $\attention_V(\vecx, \vecy)$ reach their maximum 
among all vectors $\vecy$ when $\Vp{x} = \Vs{y}$ and $\Vs{x} = \Vs{y}$ respectively. 
Secondly average attention scales the result by a factor of $\frac{1}{k}$ where $k$ is the number of such values. 
This can be counteracted by computing how many vectors were attended to and multiplying the result by that number. 
We will use the pooling functions 
\begin{align*}
	\pooling_+(X, \veca) &= \sum_{i=1}^{\vert X \vert} \avg(\veca)_i \mult X_i \in \WS{\avg} \text{ and}\\
	\pooling_\times(X, \veca) &= \bigodot_{\substack{i=1 \\ \snug{\avg(\veca)_i \neq 0}}}^{\vert X \vert} 
		\avg(\veca)_i \mult X_i \in \WP{\avg} .
\end{align*}
They require the underlying semiring to have a (partial) order and 
we assume it is compatible with the usual order of the integers. 
Then we can define \begin{align*}
	\attention_E(\vecx, \vecy) &= \Vp{x}^2 - (\Vs{y} - \Vp{x})^2 \\
	&= 2\Vp{x}\Vs{y} - \Vs{y}^2 = 2\Vs{y} \mult \Vp{x} - \Vone{x} \mult \Vssq{y} \text{ and}\\
	\attention_V(\vecx, \vecy) &= \Vs{x}^2 - (\Vs{y} - \Vs{x})^2 \\
	&= 2\Vs{x}\Vs{y} - \Vs{y}^2 = 2\Vs{x} \mult \Vs{y} - \Vone{x} \mult \Vssq{y} .
\end{align*}
This is implemented by adding two more dimensions to each vector and 
setting them to $\Vone{x} = 1$ and $\Vssq{x} = \Vs{x}^2$. 
These values can be set either by the input embedding or by the activation function of an additional transformer layer. 

To count the predecessors of each gate, we can use the numbering of the edges. 
When using average attention on a vector with predecessors $\vecx^1, \dots, \vecx^k$ 
having values $\{1, \dots, k\}$ in $\Vi{x}$, 
the result $\vecy$ of the pooling function $\pooling_+$ has 
\[\Vi{y} = \sum_{j=1}^k \frac{j}{k} = \frac{k(k+1)}{2k} = \frac{k+1}{2}.\]
Therefore we can compute the number of predecessors as $k = 2 \mult \Vi{y} - 1$. 
If this computation is also done in the first layer, 
we can multiply the intermediate values by $k$ to correct for the scaling of average attention. 

Because the definition of average attention requires multiplicative inverses of $\Natural \setminus \{0\}$, 
the following theorems only hold, when $R$ is an extension of $\Rational$. 
This also allows for using polynomials from \cref{thm:charfin_field} for the $\charfin{T}{t}$ functions. 

The modified version of $\activation_E$ is defined as 
\begin{multline*}
	\Vv{\activation_E(\vecx, \vecy, \vecz)} = 
	\big(\charfin{T}{\tinput}(\Vt{x}) + \charfin{T}{\tconst}(\Vt{x})\big) \mult \Vv{x} + \\
	\big(\charfin{T}{\toutput}(\Vt{x}) + \charfin{T}{\tplus}(\Vt{x}) + \charfin{T}{\ttimes}(\Vt{x})\big) 
	\mult (2 \mult \Vi{z} - 1) \mult \Vv{y} 
\end{multline*}
and $\activation_E(\vecx, \vecy, \vecz)_j = \vecx_j$ for $j \in \{\spiton\}$. 
For $\activation_V$, we use the previous definition while also retaining the new components. 
Finally, we define the input embedding $\finput\colon R^5 \to R^7$ as 
$\finput(\vecx) = \big( \Vs{x}, \dots, \Vv{x}, 1, \Vs{x}^2 \big)$.
The new two layer transformer is now 
\vspace{-0.5\baselineskip}
\begin{multline*}
	\Big(R, 7, R^5, \finput, 0, \attention_E, \attention_V, \attention_V, \attention_V, 
	\\
	\pooling_+, \pooling_+, \pooling_+, \pooling_\times, \activation_E, \activation_V \Big).
\end{multline*}

\begin{restatable}{theorem}{average} \label{thm:average_FAC}
	Let $R$ be a $\Integers$-ordered, commutative ring extension of $\Rational$. 
	Then for each $K \in \Natural$ there is a transformer 
	in $\transform{R}{\FCirc{R}}{\FCirc{R}}{\dpa}{\WS{\avg} \cup \WP{\avg}}$ 
	of dimension $7$ with $2K$ layers 
	that simulates any unbounded fan-in $R$-circuit of depth $\le K$ encoded with $\encoding$.
\end{restatable}

\begin{example} \label[example]{example_sim_unbounded}
	\begin{table}
		\caption{Attention values $\attention_V(\vecx, \vecy)$ for \cref{example_encoding_unbounded} 
			where $\vecx$ is given by column and $\vecy$ by row. The max of each column is bold. 
		}
		\label{tab:example_sim_unbounded_attention_V}
		\vskip 0.1in \centering
		\begin{tabular}{VVVVVVSSV}
			\toprule \attention_V & \vecx_1 & \vecx_2 & \vecx_3 & \vecx_4 & \vecx_5 & 
				\hspace{-6pt} \vecx_6, \dots, \vecx_9 \hspace*{-4pt} & 
				\hspace{-4pt} \vecx_{10}, \vecx_{11} \hspace*{-6pt} & \vecx_{12} \\ \midrule
			\vecx_1 	&   \mathbf{1} &   3 &   5 &  7 &  9 & 11 & 13 & 15 \\
			\vecx_2 	&   0 &   \mathbf{4} &   8 & 12 & 16 & 20 & 24 & 28 \\
			\vecx_3 	&  -3 &   3 &   \mathbf{9} & 15 & 21 & 27 & 33 & 39 \\
			\vecx_4 	&  -8 &   0 &   8 & \mathbf{16} & 24 & 32 & 40 & 48 \\
			\vecx_5 	& -15 &  -5 &   5 & 15 & \mathbf{25} & 35 & 45 & 55 \\
			\vecx_6 	& -24 & -12 &   0 & 12 & 24 & \mathbf{36} & 48 & 60 \\
			\vecx_7 	& -24 & -12 &   0 & 12 & 24 & \mathbf{36} & 48 & 60 \\
			\vecx_8 	& -24 & -12 &   0 & 12 & 24 & \mathbf{36} & 48 & 60 \\
			\vecx_9 	& -24 & -12 &   0 & 12 & 24 & \mathbf{36} & 48 & 60 \\
			\vecx_{10}	& -35 & -21 &  -7 &  7 & 21 & 35 & \mathbf{49} & 63 \\
			\vecx_{11}	& -35 & -21 &  -7 &  7 & 21 & 35 & \mathbf{49} & 63 \\
			\vecx_{12}	& -48 & -32 & -16 &  0 & 16 & 32 & 48 & \mathbf{64} \\
			\bottomrule
		\end{tabular}
	\end{table}

	For the circuit from \cref{example_encoding_unbounded}, 
	the attention function $\attention_V$ produces the values in 
	\cref{tab:example_sim_unbounded_attention_V}. 
\end{example}

\subsection{Simulating Bounded Fan-in}

If we only look at semi-unbounded fan-in circuits, 
the second pooling function can be eliminated from the transformer. 
To compute the values of binary multiplication gates, 
we have one attention head attend to the first edge of each gate and one to the second. 
The two values are then multiplied by the activation function. 
We define the attention function $\attention_{B(n)} \colon R^8 \times R^8 \to R$ as 
\begin{multline*}
	\attention_{B(n)} (\vecx, \vecy) 
	= \Vs{x}^2 - (\Vs{y} - \Vs{x})^2 + n^2 - (\Vi{y} - n)^2 \\
	= 2\Vs{x} \mult \Vs{y} - \Vone{x} \mult \Vssq{y} + 2n\Vone{x} \mult \Vi{y} - \Vone{x} \mult \Visq{y} .
\end{multline*}
This requires another dimension with $\Visq{x} = \Vi{x}^2$. 
When $\vecx$ has an $n$-th predecessor, 
the function reaches its maximum value only when $\Vs{y} = \Vs{x}$ and $\Vi{y} = n$. 
We can now replace the attention head using $\pooling_\times$ 
with two attention heads using $\attention_{B(1)}$ and $\attention_{B(2)}$ 
and redefine the second activation function 
\begin{align*}
	\Vv{\activation_V(\xin, \xplus, \vecx^{1}, \vecx^{2})} = \quad&\\
	\big(\charfin{T}{\tinput}(\Vt{\xin}) + \charfin{T}{\tconst}(\Vt{\xin})\big) &\mult \Vv{\xin} + \\
	\big(\charfin{T}{\toutput}(\Vt{\xin}) + \charfin{T}{\tplus}(\Vt{\xin})\big) &\mult \Vv{\xplus} + \\
	\charfin{T}{\ttimes}(\Vt{\xin}) &\mult (\Vv{x^{1}} \mult \Vv{x^{2}}) .
\end{align*}
Because the two attention heads attend only to one value, 
we need to exclude multiplication gates from the correction factor we previously applied in $\activation_E$. 
Therefore 
\begin{align*}
	\Vv{\activation_E(\vecx, \vecy, \vecz, \overline{w})} = 
	\big(\charfin{T}{\tinput}(\Vt{x}) + \charfin{T}{\tconst}(\Vt{x})\big) &\mult \Vv{x} + \\
	\big(\charfin{T}{\toutput}(\Vt{x}) + \charfin{T}{\tplus}(\Vt{x})\big) &\mult (2 \Vi{z} - 1) \mult \Vv{y} \\
	+ \; \charfin{T}{\ttimes}(\Vt{x}) &\mult \Vv{y}. 
\end{align*}
We also add the new dimension to the activation functions and the input embedding 
$\finput(\vecx) = \big( \Vs{x}, \dots, \Vv{x}, 1, \Vs{x}^2, \Vi{x}^2 \big)$. 
The new two layer transformer is now 
\vspace{-0.5\baselineskip}
\begin{multline*}
	\Big(R, 8, R^5, \finput, 0, \attention_E, \attention_V, 0, \attention_V, \attention_{B(1)}, \attention_{B(2)}, 
	\\
	\pooling_+, \pooling_+, 0, \pooling_+, \pooling_+, \pooling_+, \activation_E, \activation_V \Big).
\end{multline*}

\begin{restatable}{theorem}{averages} \label{thm:average_FSAC}
	Let $R$ be a $\Integers$-ordered ring extension of $\Rational$. 
	Then for each $K \in \Natural$ there is a transformer 
	in $\transform{R}{\FCirc{R}}{\FCirc{R}}{\dpa}{\WS{\avg}}$ 
	of dimension $8$ with $2K$ layers 
	that simulates any semi-unbounded fan-in $R$-circuit of depth $\le K$ with $\encoding$.
\end{restatable}

\begin{example} \label[example]{example_sim_semi}
	\begin{table}
		\caption{Attention values $\attention_{B(2)}(\vecx, \vecy)$ for \cref{example_encoding_unbounded} 
			where $\vecx$ is given by column and $\vecy$ by row. The max of each column is bold. 
		}
		\label{tab:example_sim_semi_attention_B2}
		\vskip 0.1in \centering
		\begin{tabular}{VVVVVVSSV}
			\toprule \attention_{B(2)} & \vecx_1 & \vecx_2 & \vecx_3 & \vecx_4 & \vecx_5 & 
				\hspace{-6pt} \vecx_6, \dots, \vecx_9 \hspace*{-4pt} & 
				\hspace{-4pt} \vecx_{10}, \vecx_{11} \hspace*{-6pt} & \vecx_{12} \\ \midrule
			\vecx_1 	&   \mathbf{1} &  3 &   5 &  7 &  9 & 11 & 13 & 15 \\
			\vecx_2 	&   0 &  \mathbf{4} &   8 & 12 & 16 & 20 & 24 & 28 \\
			\vecx_3 	&  -3 &  3 &   \mathbf{9} & 15 & 21 & 27 & 33 & 39 \\
			\vecx_4 	&  -8 &  0 &   8 & \mathbf{16} & 24 & 32 & 40 & 48 \\
			\vecx_5 	& -15 & -5 &   5 & 15 & 25 & 35 & 45 & 55 \\
			\vecx_6 	& -21 & -9 &   3 & 15 & 27 & 39 & 51 & 63 \\
			\vecx_7 	& -20 & -8 &   4 & \mathbf{16} & \mathbf{28} & \mathbf{40} & 52 & 64 \\
			\vecx_8 	& -21 & -9 &   3 & 15 & 27 & 39 & 51 & 63 \\
			\vecx_9 	& -24 & -1 &   0 & 12 & 24 & 36 & 48 & 60 \\
			\vecx_{10}	& -32 & -1 &  -4 & 10 & 24 & 38 & 52 & 66 \\
			\vecx_{11}	& -31 & -1 &  -3 & 11 & 25 & 39 & \mathbf{53} & \mathbf{67} \\
			\vecx_{12}	& -45 & -2 & -13 &  3 & 19 & 35 & 51 & \mathbf{67} \\
			\bottomrule
		\end{tabular}
	\end{table}

	For the circuit from \cref{example_encoding_unbounded}, 
	the attention function $\attention_{B(2)}$ produces the values in 
	\cref{tab:example_sim_semi_attention_B2}. 
	Here we see that our method of attention only works for gates that have a second predecessor, 
	as for some others, we get multiple maximal scores. 
	This is not an issue, because for other gates the values from the corresponding attention head 
	are never used in the activation functions. 
\end{example}

If we only look at bounded fan-in circuits, we can apply the same construction to addition and output gates and 
also remove the attention head for predecessor counting. 
As this makes all heads attend to at most one relevant vector, 
we can even replace average with leftmost or rightmost attention. 

\begin{restatable}{theorem}{averagen} \label{thm:average_FNC}
	Let $R$ be a $\Integers$-ordered ring extension of $\Rational$. 
	Then for each $K \in \Natural$ there is a transformer 
	in $\transform{R}{\FCirc{R}}{\FCirc{R}}{\dpa}{\WS{\hardleft}}$ 
	of dimension $8$ with $2K$ layers 
	that simulates any bounded fan-in $R$-circuit of depth $\le K$ encoded with $\encoding$.
\end{restatable}

\subsection{Simulating Circuit Extensions}

Using the attention functions $\attention_{B(3)}, \attention_{B(4)}, \dots$ 
we can fetch the values from an arbitrary number of predecessors. 
This works as long as the fan-in of the gates is bounded by the number of attention heads. 
So for any set $B$ of functions $f \colon R^k \to R$, 
we can simulate circuits over the basis $(R, B)$ 
if we allow activation functions from $\FCirc{R}[B]$. 
This requires an expansion of the encoding $\encoding$ to include the gate types for $B$. 

\begin{restatable}{theorem}{averagee} \label{thm:average_extended}
	Let $R$ be a $\Integers$-ordered ring extension of $\Rational$, 
and	$B$ a finite set of functions $f \colon R^k \to R$. 
	Then for each $K \in \Natural$ there is a transformer 
	in $\transform{R}{\FCirc{R}}{\FCirc{R}[B]}{\dpa}{\WS{\avg}}$ 
	with $2K$ layers 
	that simulates any semi-unbounded fan-in $R$-circuit 
	of depth $\le K$ encoded with $\encoding$, 
	which also uses functions from $B$.  
\end{restatable}

The $\avg$ function is non-continuous and we can use $\WS{\avg}$ to compute the $\sign$ function. 
By constructing an attention function that attends 
to one specific vector if $\Vv{x} < 0$ and to two vectors if $\Vv{x} = 0$, 
we can use the difference in the resulting value to compute the $\zero$ function. 
If we add another dimension containing only values $0$ and $1$, 
we can construct an attention head that distinguishes between positive and negative values. 
Together with $\zero$ this is sufficient to compute $\sign$. 

\begin{theorem} \label{thm:simulate_sign}
	Let $R$ be a $\Integers$-ordered ring extension of $\Rational$. 
	Then for each $K \in \Natural$ there is a transformer 
	in $\transform{R}{\FCirc{R}}{\FCirc{R}}{\dpa}{\WS{\avg}}$ 
	of dimension $9$ with $2K$ layers 
	that simulates any $\FSAC{R}{0}[\sign]$-circuit of depth $\le K$ encoded with $\encoding$.
\end{theorem}

\section{Conclusion}

In this paper, we established a correspondence between 
generalized average attention transformers using circuits as activation functions and 
the arithmetic circuit class $\FSAC{R}{0}[\sign]$. 
Over any semiring $R$, these transformers only compute functions in $\FSAC{R}{0}[\sign]$ and 
if $R$ is a $\Integers$-ordered ring extension of $\Rational$, 
there is a transformer for each $\FSAC{R}{0}[\sign]$-family that can simulate all of its circuits encoded with $\encoding$. 
Unfortunately, this is not possible while only allowing feed-forward networks as activation functions, 
because those cannot perform multiplication of two inputs. 

This does not imply that any $\FSAC{R}{0}[\sign]$-function is also computed by a transformer. 
In fact, this would be impossible, as circuit families are a non-uniform model of computation. 
If we used transformers to generate longer sequences of vectors, similar to how they are used in natural language processing, 
we could also compute the encoding of a circuit from just the input vector. 
Our result would then lead to an equivalence between circuit families that are uniformly described by a transformer and 
the functions computed by such transformers. 

\section*{Impact Statement}

This paper presents work whose goal is to advance the field of machine learning theory. 
There are many potential societal consequences of our work, 
none of which we feel must be specifically highlighted here.

\bibliography{references.bib}

\newpage

\appendix

\section*{Appendix}
\lagrange* \label{app:lagrange}
\begin{proof}
	The Lagrange polynomial 
	\[\fpoly{A}{a} (x) = \prod_{b \in A \setminus \{a\}} \frac{x - b}{a - b}\]
	has the desired property. 
	It can be computed by a circuit using constants from 
	$A \cup \{(a - b)^{-1} \mid a, b \in A, a \neq b\}$. 
\end{proof}

\subsection*{Proofs of \cref{sec:simulating_circuits}}

\newcommand{\sigh}{c}

\generalized*
\begin{proof}
	Let $C$ be an unbounded fan-in arithmetic circuit over $R$ of depth $k$ with $m$ inputs and $\vecu \in R^m$. 
	We assume that $R$ is commutative to avoid values of multiplication gates being dependent on the order of the input vector. 
	Let $\ell\colon (R^5)^* \to (R^5)^*$ be the function computed by the two layers described in \cref{sec:sim_generalized}. 
	We construct a transformer by simply copying those layers $k$ times and 
	show that in its output, each vector $\vecx$ 
	has $\Vv{x} = f_{C, g}(\vecu)$ where $g$ is the gate with index $\Vs{x}$. 
	This is done by induction over path lengths in $C$. 
	
	The input to the first layer is equal to $\encoding(C, \vecu)$, 
	because there is no positional embedding. 
	Therefore, all input and constant gates already have the correct value. 

	Assume that each vector $\vecx$ in $\ell^k (\encoding(C, \vecu))$ whose corresponding gate $g$ 
	has no path of length $> k$ from an input or constant gate has $\Vv{x} = f_{C, g}(\vecu)$. 
	Now let $\vecx$ be any vector in $\ell^k (\encoding(C, \vecu))$ whose corresponding gate $g$ 
	has no such path of length $> k+1$ and 
	let $\vecy$ be the vector corresponding to $\vecx$ in $\ell^{k+1} (\encoding(C, \vecu))$. 
	\begin{noseptemize}
		\item If $g$ is an input or constant gate, we get $\vecy = \vecx = f_{C, g}(\vecu)$, 
			because $\vecy = \activation_V(\activation_E(\vecx, \vecz^1, \vecz^2), \vecz^3, \vecz^4) = \vecx$ 
			for any $\vecz^1, \dots, \vecz^4 \in R^5$. 
		\item If $g$ is an output or addition gate, let 
			$\vecz^E$ be the result of the attention head using $\attention_E$ on $\vecx$, 
			$\vecx' = \activation_E(\vecx, \vecz^1, 0)$ and 
			$\vecz^+$ be the result of the attention head using $\attention_V$ and $\pooling_+$ on $\vecx'$. 
			Then $\attention_E$ makes $\vecx$ attend to exactly those vectors $\overline{w}$ which corresponds to 
			the predecessor $p$ of $g$ that has index $\Vp{x}$. 
			At least one such vector must exist, if the input was a circuit encoded with $\encoding$. 
			From our assumption, we get $\Vv{x'} = \Vv{w} = f_{C, p}(\vecu)$. 
			In the second layer, $\attention_V$ makes $\vecx'$ attend to all vectors that correspond to $g$ including itself. 
			Therefore, \[\Vv{z^+} = \sum_{p \in \neighbourin(g)} f_{C, p}(\vecu) = f_{C, g}(\vecu) .\] 
			Then for any $\vecz^\times$, $\vecy = \activation_V(\vecx', \vecz^+, \vecz^\times)$ has 
			$\Vv{y} = \Vv{z^+} = f_{C, g}(\vecu)$ and $y_j = x'_j = x_j$ for $j \in \{\spit\}$. 
		\item If $g$ is a multiplication gate, we analogously get 
			\[\Vv{y} = \prod_{p \in \neighbourin(g)} f_{C, p}(\vecu) .\]
	\end{noseptemize}
	So by induction we get $\Vv{x} = f_{C, g}(\vecu)$ for all $\vecx$ corresponding to $g$ 
	in $\ell^K (\encoding(C, \vecu))$ if $\depth(C) \le K$. 
	Since $\ell^K (\encoding(C, \vecu))$ is the output of a transformer which applies the two layers $K$ times, 
	this proves \cref{thm:generalized_id}. 
\end{proof}

\average*
\begin{proof}
	This proof is very similar to the one of \cref{thm:generalized_id}. 
	For any $\vecx$, the two new attention functions $\attention_E(\vecx, \vecy)$ and $\attention_V(\vecx, \vecy)$ 
	reach their maximum values $\Vp{x}^2$ and $\Vs{x}^2$ only when the previous attention functions output $1$. 
	If the input is a circuit encoded with $\encoding$, 
	such a $\vecy$ always exists for output, $+$ and $\times$ gates. 
	So in those cases, the attention heads attend to exactly the same vectors as before and 
	because of the scaling factor, the component $\Vv{z}$ of their output does not change. 
	For input and constant gates, the output of the attention heads is ignored. 
	The remainder of the proof is the same as for \cref{thm:generalized_id}. 
\end{proof}

\averages*
\begin{proof} \label{pro:average_FSAC}
	This proof is very similar to the one of \cref{thm:average_FAC}. 
	All steps for input, constant, output and $+$ gates are the same. 
	In the induction step, if $g$ is a multiplication gate, 
	$\attention_E$ works in the same way as before, 
	but the resulting value is not scaled by a factor in $\activation_E$. 
	Then the attention functions $\attention_{B(1)}$ and $\attention_{B(2)}$ 
	attend only to the vectors $\vecy^1$ and $\vecy^2$ corresponding to $g$ that have $\Vi{y^1} = 1$ and $\Vi{y^2} = 2$. 
	If the input is a semi-unbounded fan-in circuit encoded with $\encoding$, 
	all $\times$ gates must have exactly two predecessors, so these vectors must exist. 
	For the resulting vector $\vecz = \activation_V(\vecx', \vecx^+, \vecy^1, \vecy^2)$ we get 
	$\Vv{z} = \Vv{y^1} \mult \Vv{y^2} = f_{C, g}(\vecu)$. 
\end{proof}

\averagen*
\begin{proof} \label{pro:average_FNC}
	We redefine the activation functions: 
	\begin{align*}
		\Vv{\activation_V(\xin, \vecx^{1}, \vecx^{2})} = \quad&\\
		\big(\charfin{T}{\tinput}(\Vt{\xin}) + \charfin{T}{\tconst}(\Vt{\xin})\big) &\mult \Vv{\xin} + \\
		\charfin{T}{\toutput}(\Vt{\xin}) &\mult \Vv{x^{1}} + \\
		\charfin{T}{\tplus}(\Vt{\xin}) &\mult (\Vv{x^{1}} + \Vv{x^{2}}) + \\
		\charfin{T}{\ttimes}(\Vt{\xin}) &\mult (\Vv{x^{1}} \mult \Vv{x^{2}}) 
	\end{align*}
	\begin{multline*}
		\Vv{\activation_E(\vecx, \vecy, \vecz)} = 
		\big(\charfin{T}{\tinput}(\Vt{x}) + \charfin{T}{\tconst}(\Vt{x})\big) \mult \Vv{x} + \\
		\big(\charfin{T}{\toutput}(\Vt{x}) + \charfin{T}{\tplus}(\Vt{x}) + \charfin{T}{\ttimes}(\Vt{x})\big) 
		\mult \Vv{y} 
	\end{multline*}
	The new two layer transformer is now 
	\vspace{-0.5\baselineskip}
	\begin{multline*}
		\Big(R, 8, R^5, \finput, 0, \attention_E, 0, \attention_{B(1)}, \attention_{B(2)}, 
		\\
		\pooling_+, 0, \pooling_+, \pooling_+, \activation_E, \activation_V \Big).
	\end{multline*}
	We can also replace $\pooling_+$ with the function 
	\[\pooling_{\hardleft}(X, \veca) = \sum_{i=1}^{\vert X \vert} \hardleft(\veca)_i \mult X_i \in \WS{\hardleft} .\]
	The proof for $\pooling_+$ follows the same steps for addition 
	as the proof of \cref{thm:average_FSAC} does for multiplication. 
	In addition, the result from $\attention_{B(2)}$ is ignored for output gates, 
	but $\attention_{B(1)}$ still works in the same way. 

	Because $\attention_{B(1)}$ and $\attention_{B(2)}$ attend to exaclty one vector in all relevant cases, 
	using $\pooling_{\hardleft}$ makes no difference for the corresponding attention heads. 
	The third head might attend to multiple vectors at a time, 
	but, if the induction hypothesis holds, all of those share the same $\Vv{y}$. 
	Since $\Vv{y}$ is now the only component used in $\activation_E$, 
	$\pooling_{\hardleft}$ produces the same results as $\pooling_+$. 
\end{proof}

\averagee*
\begin{proof}
	Let $m$ be the maximum arity of functions in $B$. 
	We modify the transformer from \cref{thm:average_FSAC} 
	by adding attention heads for $\attention_{B(3)}, \dots, \attention_{B(m)}$ using $\pooling_+$ 
	in the same way as for $\attention_{B(1)}$ and $\attention_{B(2)}$. 
	If the encoding uses constants $T_B = \{t_b \mid b \in B\}$ for the new functions, 
	we can define $T' = T \cup T_B$ and use the functions $\charfin{T'}{t}$ to define 
	$\activation_V \colon (R^d)^{m+2} \to R^d$ as 
	\begin{align*}
		\Vv{\activation_V(\xin, \xplus, \vecx^{1}, \dots, \vecx^{m})} = \quad&\\
		\big(\charfin{T'}{\tinput}(\Vt{\xin}) + \charfin{T'}{\tconst}(\Vt{\xin})\big) &\mult \Vv{\xin} + \\
		\big(\charfin{T'}{\toutput}(\Vt{\xin}) + \charfin{T'}{\tplus}(\Vt{\xin})\big) &\mult \Vv{\xplus} + \\
		\charfin{T'}{\ttimes}(\Vt{\xin}) &\mult (\Vv{x^{1}} \mult \Vv{x^{2}}) + \\
		\sum_{b \in B} \charfin{T'}{t_b} &\mult b(\Vv{x^{1}}, \dots, \Vv{x^{k_b}})
	\end{align*}
	where $k_b$ is the arity of $b$. 
\end{proof}

\begin{proof}[Proof of \cref{thm:simulate_sign}]
	To compute the $\zero$ function we use two attention functions $\attention_{z+}, \attention_{z-} \in \dpa$
	defined as 
	\begin{alignat*}{2}
		\attention_{z\pm}(\vecx, \vecy) &= \pm \Vv{x} \mult \Vs{y} &&+ 
			\frac{9}{4} - (\Vs{y} - \frac{3}{2})^2 + \frac{1}{4} - (\Vi{y} - \frac{1}{2})^2 \\
		&= \pm \Vv{x} \mult \Vs{y} &&+ \frac{3}{2} \Vone{x} \mult \Vs{y} - \Vone{x} \mult \Vssq{y} \\
		&	&&+ \frac{1}{2} \Vone{x} \mult \Vi{y} - \Vone{x} \mult \Visq{y} .
	\end{alignat*}
	Because every non-empty circuit has at least two nodes and 
	every valid encoding of a circuit has only integer values in $\Vs{x}$ and $\Vi{x}$, 
	the term $-(\Vs{y} - \frac{3}{2})^2$ reaches its maximum only when $\Vs{y} = 1$ or $\Vs{y} = 2$ and 
	the term $-(\Vi{y} - \frac{1}{2})^2$ reaches its maximum only when $\Vi{y} = 0$ or $\Vi{y} = 1$. 
	There are two cases for each of the two gates with indices $1$ and $2$: 
	\begin{itemize}
		\item If the gate is an input or constant gate, 
			it has only one corresponding vector $\vecy$ with $\Vi{y} = 0$. 
		\item Otherwise, the gate may have multiple corresponding vectors, 
			but it has exactly one with $\Vi{y} = 1$. 
	\end{itemize}
	In both cases, there are exactly two vectors for which the term 
	$- (\Vs{y} - \frac{3}{2})^2 - (\Vi{y} - \frac{1}{2})^2$ is maximal. 
	Let those be $\vecy$ and $\vecz$ such that $\Vs{y} = 1$ and $\Vs{z} = 2$. 
	Then for any vector $\vecx$, we get 
	$\attention_{z+}(\vecx, \vecz) - \attention_{z+}(\vecx, \vecy) = (\Vs{z} - \Vs{y}) \mult \Vv{x} = \Vv{x}$ 
	and $\attention_{z-}(\vecx, \vecz) - \attention_{z-}(\vecx, \vecy) = -\Vv{x}$. 
	When these attention functions are used in an attention head with $\pooling_+$ 
	resulting in vectors $\vecu^+$ and $\vecu^-$, there are three cases for $\Vs{u^+}$ and $\Vs{u^-}$: 
	\begin{itemize}
		\item If $\Vv{x} = 0$, $\vecx$ attends only to $\vecy$ and $\vecz$. 
			Therefore, $\Vs{u^+} = \Vs{u^-} = \frac{1 + 2}{2} = \frac{3}{2}$. 
		\item If $\Vv{x} < 0$, $\attention_{z+}$ makes $\vecx$ attend only to $\vecy$, 
			because $\attention_{z+}(\vecx, \vecy) > \attention_{z+}(\vecx, \vecz)$. 
			Therefore, $\Vs{u^+} = 1$. 
		\item If $\Vv{x} > 0$, $\attention_{z-}$ makes $\vecx$ attend only to $\vecy$, 
			because $\attention_{z-}(\vecx, \vecy) > \attention_{z-}(\vecx, \vecz)$. 
			Therefore, $\Vs{u^-} = 1$. 
	\end{itemize}
	Any other vector $\overline{w}$ cannot have an effect here, 
	because we get $-\Vv{x} \mult \Vs{w} \le -\Vv{x} \mult \Vs{y}$ for $\Vv{x} > 0$ and 
	$\Vv{x} \mult \Vs{w} \le \Vv{x} \mult \Vs{y}$ for $\Vv{x} < 0$, as $\Vs{w} \ge \Vs{y}$. 
	We can now compute the $\zero$ function from $\vecu^+$ and $\vecu^-$ as follows: 
	\[\zero(\Vv{x}) = 4 \mult (\Vs{u^+} - 1) \mult (\Vs{u^-} - 1)\]

	To compute the $\sign$ function, we define a new dimension $\Vbin{x} = \charfin{T}{\toutput}(\Vt{x})$. 
	Because the transformer does not need to compute anything if there are no output gates and 
	if there is an output gate there must also be a non-output gate, 
	without loss of generality there is a vector with $\Vbin{x} = 1$ and a vector with $\Vbin{x} = 0$. 
	We define another attention function 
	\[\attention_{\sign}(\vecx, \vecy) = \Vv{x} \mult \Vbin{y}\]
	If $\Vv{x} = 0$, it attends uniformly to all vectors, 
	if $\Vv{x} < 0$, it attends only to vectors with $\Vbin{y} = 0$ and 
	if $\Vv{x} > 0$, it attends only to vectors with $\Vbin{y} = 1$. 
	Therefore, if an attention head using $\attention_{\sign}$ and $\pooling_+$ 
	produces the vector $\vecu$ from $\vecx$, we get 
	\[\sign(\Vv{x}) = (1 - \zero(\Vv{x})) \mult (2 \mult \Vbin{u} - 1) .\]

	We add this term to $\activation_V$ and define a new input embedding
	$\finput(\vecx) = \big( \Vs{x}, \dots, \Vv{x}, 1,\allowbreak \Vs{x}^2, \Vi{x}^2, \charfin{T}{\toutput}(\Vt{x}) \big)$. 
	The resulting two layer transformer is 
	\vspace{-0.5\baselineskip}
	\begin{multline*}
		\Big(R, 9, R^5, \finput, 0, 
		\attention_E, \attention_V, 0, 0, 0, 0, \\
		\attention_V, \attention_{B(1)}, \attention_{B(2)}, \attention_{z+}, \attention_{z-}, \attention_{\sign}, 
		\pooling_+, \pooling_+, 0, 0, 0, 0, \\
		\pooling_+, \pooling_+, \pooling_+, \pooling_+, \pooling_+, \pooling_+, 
		\activation_E, \activation_V \Big).
	\end{multline*}
	As before, it can be concatenated with itself to simulate circuit families of constant depth which use $\sign$ gates. 
\end{proof}

\end{document}